\documentclass[sigconf]{acmart}
\AtBeginDocument{%
  \providecommand\BibTeX{{%
    \normalfont B\kern-0.5em{\scshape i\kern-0.25em b}\kern-0.8em\TeX}}}

\usepackage{subcaption}
\usepackage{amsmath,amsfonts, amsthm}
\usepackage{graphicx}
\usepackage{textcomp}
\usepackage{xcolor}
\usepackage{graphicx}
\usepackage{tabularx,ragged2e}
\usepackage[utf8]{inputenc}
\usepackage{kotex}
\usepackage{color}
\usepackage{url}
\usepackage{footnote}
\usepackage{float}
\usepackage{comment}
\usepackage{caption}
\usepackage{algpseudocode}
\usepackage{soul}
\usepackage{url}
\usepackage{siunitx}
\usepackage[resetcount,linesnumbered,algoruled,boxed,lined]{algorithm2e}
\usepackage{multirow}
\usepackage{enumitem}
\usepackage[toc,page]{appendix}
\usepackage{hyperref}
 \usepackage[normalem]{ulem}

\newcommand{\kDS}{{\textsf{$p$DSS}}}

\SetCommentSty{mycommfont}
\SetKwInOut{Parameter}{parameter}
\SetKwFor{For}{for (}{) $\lbrace$}{$\rbrace$}

\newcolumntype{C}{>{\Centering\arraybackslash}X} 

\newtheorem{problemDefinition}{Problem Definition}
\newtheorem{definition}{Definition}
\newtheorem{example}{Example}
\newtheorem{property}{Property}
\newtheorem{remark}{Remark}
\newtheorem{theorem}{Theorem}

\newcommand{\spara}[1]{\smallskip\noindent{\bf #1}}

\setcopyright{acmcopyright}
\copyrightyear{2018}
\acmYear{2018}
\acmDOI{XXXXXXX.XXXXXXX}

\acmConference[Conference acronym 'XX]{Make sure to enter the correct
  conference title from your rights confirmation emai}{June 03--05,
  2018}{Woodstock, NY}
\acmPrice{15.00}
\acmISBN{978-1-4503-XXXX-X/18/06}

\begin{document}

\title{OCSM : Finding Overlapping Cohesive Subgraphs with Minimum Degree}

\author{Junghoon Kim}
\affiliation{%
  \institution{Nanyang Technological University}
  \country{Singapore}}
\email{junghoon001@e.ntu.edu.sg}

\author{Sungsu Lim}
\affiliation{%
  \institution{Chungnam National University}
  \country{South Korea}}
\email{sungsu@cnu.ac.kr}

\author{Jungeun Kim}
\affiliation{%
  \institution{Kongju National University}
  \country{South Korea}
}
\email{jekim@kongju.ac.kr}

\renewcommand{\shortauthors}{Kim et al.}

\begin{abstract}
Cohesive subgraph discovery in a network is one of the fundamental problems and investigated for several decades. In this paper, we propose the \underline{O}verlapping \underline{C}ohesive \underline{S}ubgraphs with \underline{M}inimum degree (OCSM) problem which combines three key concepts for OCSM : (i) edge-based overlapping, (ii) the minimum degree constraint, and (iii) the graph density. 
To the best of our knowledge, this is the first work to identify overlapping cohesive subgraphs with the  minimum degree by incorporating the graph density. Since the OCSM problem is NP-hard, we propose two algorithms: advanced peeling algorithm and seed-based expansion algorithm. Finally, we show the experimental study with real-world networks to demonstrate the effectiveness and efficiency of our proposed algorithms.
\end{abstract}



\maketitle

\section{Introduction}\label{sec:introduction}
\subsection{Motivation}\label{sec:motivation}
With recent rapid and important developments in mobile and IT technology, many people have started using Social Networking Services (SNSs) all the time and everywhere.
Considering the vast number of social networks, the mining of \textit{cohesive subgraphs} in a social network has been widely studied~\cite{seidman1983network,cohen2008trusses} even if there is no formal definition. Normally, a cohesive subgraph is considered to be a group of users that are highly connected with each other.
Recently, many cohesive subgraph models are proposed including $k$-core~\cite{seidman1983network}, $(\alpha,\beta)$-core~\cite{he2021exploring} $k$-clique~\cite{tsourakakis2015k}, and $k$-truss~\cite{cohen2008trusses,zheng2017finding}. 
Among them, \textit{$k$-core}~\cite{seidman1983network} is the most popular and widely used model owing to its simple and intuitive structure. The definition of $k$-core~\cite{seidman1983network} is as follows: given a graph $G$ and a positive integer $k$, a $k$-core, denoted as $D$, is a maximal subgraph of which all nodes in the subgraph have at least $k$ neighbor nodes in $D$. $k$-core has many applications, such as community search problem~\cite{wu2015robust,fang2020survey,cui2014local,kim2020densely}, user engagement maximization problem~\cite{bhawalkar2015preventing, zhang2017olak, linghu2020global}. Furthermore, it is known that the $k$-core can play a role as a subroutine for much harder problems~\cite{khaouid2015k,zhu2020community} and can be utilized in different networks~\cite{bai2020efficient}. 

Even if $k$-core is widely used and has many applications, $k$-core intrinsically suffers from several limitations due to its definition: (i) it returns a relatively large solution, especially when the value of $k$ is small, \textit{i.e.,} it may contain loosely connected nodes; (ii) it always returns a disjoint result, \textit{i.e.,} it cannot reveal overlapping structures. 

The reason for the large solution of $k$-core is its maximality constraint. In an Amazon dataset~\cite{yang2015defining}, when we apply $3$-core, $98\%$ of the nodes belong to a single giant connected component. Similarly, in a Youtube dataset, $99.9\%$ of the nodes belong to a single giant connected component. 

Many studies show that people in a real social network can be intrinsically characterized by multiple cluster memberships~\cite{xie2012towards,ding2022ceo}. Kelley et al.~\cite{kelley2011handbook} shows that membership overlap is a significant characteristic of many real-world social networks. 
We can easily notice that a cohesive subgraph structure can overlap. 
In real life, people can belong to multiple groups, such as a dance club, table tennis club, family, graduate student association, and so on and can be engaged in all these groups. It indicates that the cohesive subgraphs can overlap. Therefore, in our paper, we focus on finding overlapping cohesive subgraphs. 

\subsection{OCSM Problem}\label{sec:ocsm}
To handle the problem of $k$-core, in this paper, we propose an \underline{O}verlapping \underline{C}ohesive \underline{S}ubgraph with \underline{M}inimum degree (\textsf{OCSM}) problem by resolving the limitation of the $k$-core.  

At first, to incorporate the overlapping structure into the cohesive subgraph discovery problem, we use a line-graph~\cite{evans2009link} which represents the adjacencies between edges of a network. This line-graph helps in identifying the latent structures by changing the perspective from the node-level to the edge-level. 
However, not only the original line-graph~\cite{evans2009link} but also its extension, called the link-space graph~\cite{lim2014linkscan}, suffer from efficiency and effectiveness problems. Thus, we propose a \textit{link-skein graph}, which is a subset of the link-space graph with the edges which form high-order structures (e.g., triangles) in the original graph in order to preserve meaningful information, while significantly improving the efficiency. 

Next, to avoid finding a large solution with loosely connected nodes, we incorporate the graph density~\cite{feige2001dense} into the cohesive subgraph discovery problem. By maximizing the graph density of the cohesive subgraphs, we can achieve more cohesive subgraphs as a result. As we aim to find multiple overlapping cohesive subgraphs, we newly define a \textit{link-density}, which is an extension of the graph density for link-skein graphs. 

Table~\ref{tab:compare} briefly compares the result of {\textsf{OCSM}} with $k$-core and the densest subgraph (DS) discovery problem. 
Only {\textsf{OCSM}} can retrieve the top $t$ overlapping subgraphs while satisfying the minimum degree constraint. Furthermore, the {\textsf{OCSM}} problem is not trivial as it is proven to be NP-hard.

\begin{table}[ht]
\centering
\caption{Comparison of the {\textsf{OCSM}}, $k$-core, and densest subgraph discovery (DS).}
\label{tab:compare}
\begin{tabular}{c|c|c|c}
\hline
      &{\textsf{OCSM}}[\textcolor{blue}{this work}]     & $k$-core~\cite{seidman1983network}     & DS~\cite{feige2001dense}       \\ \hline \hline
Constraint           & min. degree            & min. degree          & connectivity           \\ \hline
Overlap            & Yes                & No               & No                \\ \hline
Result             & top $t$              & a set of            & the densest            \\
                & dense subgraphs          & nodes             & subgraph             \\ \hline
Objective            & max. density                & max. size               & max. density                \\ \hline
\hline
\end{tabular}
\end{table}

To solve the {\textsf{OCSM}} problem, we propose two heuristic algorithms: (i) an advanced peeling algorithm (\textsf{APA}) and (ii) a seed-based expansion algorithm (\textsf{SEA}). 
The high-level idea of the proposed algorithms is as follows. 
The first is a top-down approach which iteratively deletes a set of nodes to maximize the link-density while satisfying the degree constraint. 
In contrast, the latter is a bottom-up approach which identifies the densely connected seed nodes, and then, iteratively adds a set of nodes to satisfy the degree constraint. 
These procedures are iteratively repeated until the top $t$ subgraphs are found.

\begin{figure}[ht]
 \centering
 \begin{subfigure}[h]{0.99\linewidth}
 \includegraphics[width=0.99\linewidth]{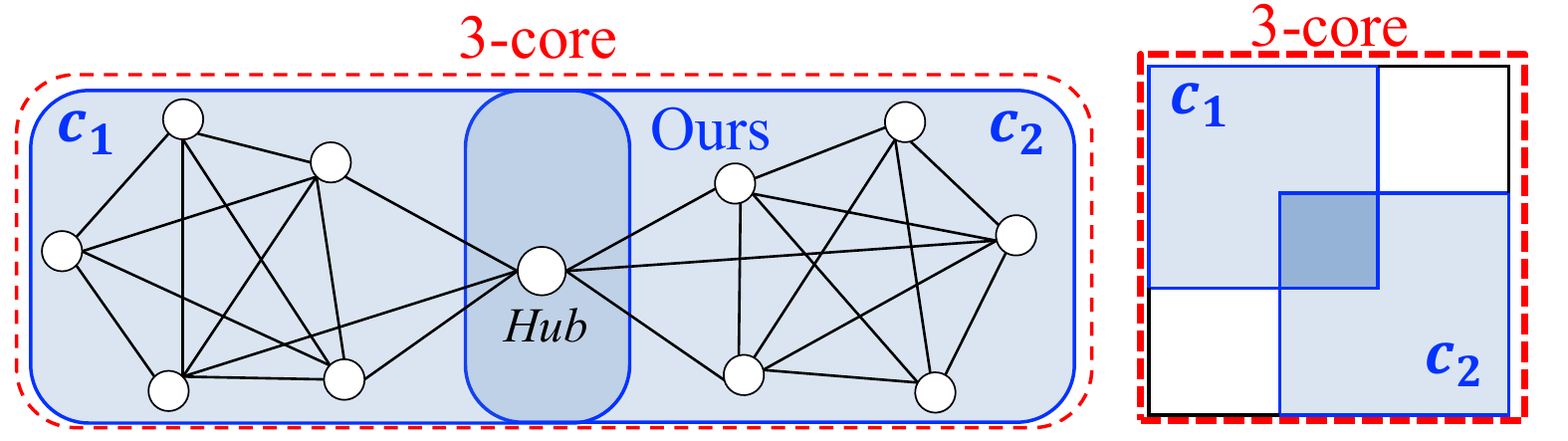}
  \caption{Toy network \#1}
  \label{fig:toy1}
 \end{subfigure} 
 \begin{subfigure}[h]{0.99\linewidth} 
 \includegraphics[width=0.99\linewidth]{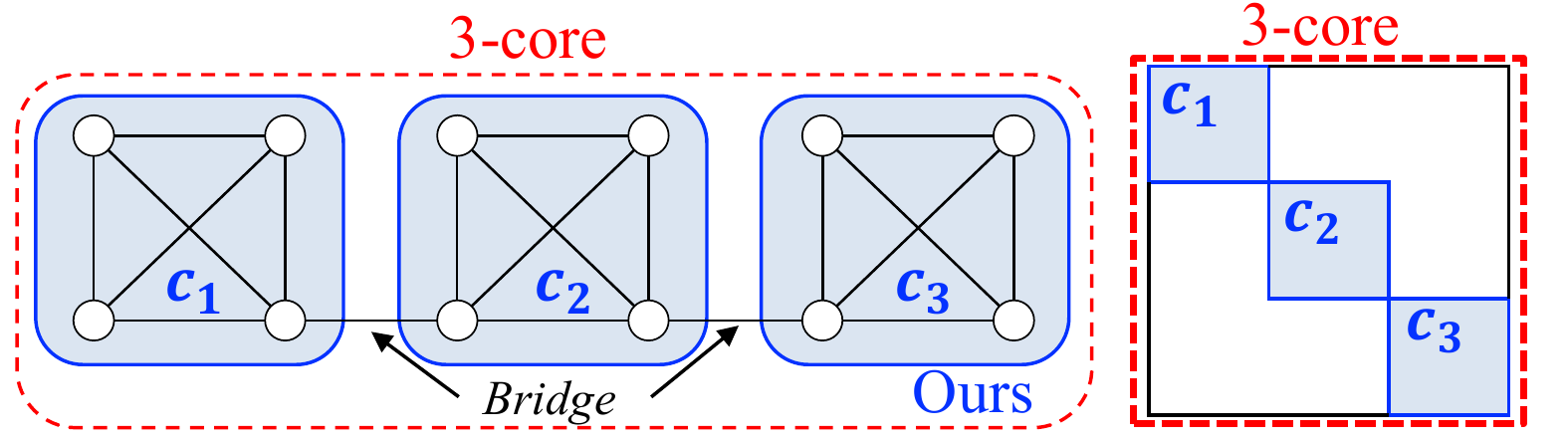}
  \caption{Toy network \#2}
  \label{fig:toy2}
 \end{subfigure}
 \caption{Motivating example using toy networks}
 \label{fig:toys}
\end{figure}

\begin{example}

Figure~\ref{fig:toys} shows the subgraph mining results obtained by the {\textsf{OCSM}} problem and a $k$-core with two toy networks. To present the overall cohesive subgraph structure, we put the high-level community structure on the right-side of the figures. 
First, Figure~\ref{fig:toy1} shows a simple network with $11$ nodes and two cohesive subgraphs with one overlapping hub node. When we use $k$-core, it fails to find two dense cohesive subgraphs. This is because the hub node has a high degree and connects the two cohesive subgraphs, even if the edges of the hub node are not related to each other. Even though the density of each cohesive subgraph is larger than that of the whole graph, $k$-core always identifies the whole graph as a result. Thus, it also fails to find a hub (overlapping) node. Note that {\textsf{OCSM}} can retrieve two cohesive subgraphs that overlap at a node. 
Next, Figure~\ref{fig:toy2} shows the case in which $k$-core cannot identify cohesive subgraphs and cannot specify the number of cohesive subgraphs. We notice that there are two bridge edges between the cohesive subgraphs. These bridge edges connect two cohesive subgraphs, then, $k$-core returns large subgraphs. Even if the user already knows the number of cohesive subgraphs in advance, $k$-core cannot incorporate this information. Note that the {\textsf{OCSM}} can identify three cohesive subgraphs since our approach lessens the influence of the bridge edges.
\end{example}

\subsection{Key Contributions}\label{sec:contributions}
\begin{itemize}
  \item \textbf{Problem Significance}: We formally define the {\textsf{OCSM}} problem. To the best of our knowledge, this is the first work to find top $t$ densely connected overlapping subgraphs discovery with a minimum degree constraint.
  \item \textbf{Solution}: We theoretically show that our problem is NP-hard and propose two heuristic algorithms for addressing the {\textsf{OCSM}} problem.
  \item \textbf{Extensive Evaluations}: We conduct extensive experiments on real-world datasets to check the efficiency and effectiveness. Furthermore, an interesting case study shows that our solution successfully discovers densely connected overlapping subgraphs.
\end{itemize}

\section{Problem Statement}\label{sec:problemStatement}
In this section, we formally introduce our problem and its hardness. We assume that all graphs considered in this work are simple and undirected. 
Given a subset of nodes $V' \subseteq V$, we denote $G[V'] = (V', E[V'])$ the subgraph of $G=(V,E)$ induced by $V'$, i.e., $E[V'] = \{\{i,j\}\in E| i, j\in V'\}$. The basic notations are summarized in Table~\ref{tab:term}.

\begin{table}[t]
\centering 
\caption{Basic notation}
\vspace*{-0.3cm} 
\label{tab:term}
\begin{tabular}{c|cl}
\hline
Notation      & \multicolumn{1}{c}{Definition}        &  \\ \hline \hline
$G=(V,E)$     & an original graph               &  \\ 
$L(G)=(V_{L(G)}, E_{L(G)})$ & the link-skein graph of $G$ &  \\ 
$e_{i,j}$      & an edge of the nodes $i$ and $j$ in $G$ & \\
$v_{i,j}$      & a node generated by $e_{i,j}$ in $L(G)$ & \\
$R(H)$      & a set of nodes in $G$ from $H\subseteq V_{L(G)}$                  &  \\ 
$E_{L(G)}$     & edges of link-skein graph of $G$ & \\
$N(v)  $      & a set of neighbor nodes of node $v$  &  \\
$w(v, u)$      & the edge weight of $v$ and $u$ in $L(G)$  &  \\
$\delta(G)$     & min. degree of graph $G$          &  \\
$\beta(H)$     & min. occurrence of $H \subseteq V_{L(G)}$          &  \\
$\gamma(H)$     & link-density of $H\subseteq V_{L(G)}$             &  
\\ \hline \hline
\end{tabular}
\end{table}

\subsection{Link-Space and Link-Skein Graphs}

We first introduce the link-space~\cite{lim2014linkscan,kim2018linkblackhole} and link-skein graphs, which have several benefits for the overlapping cohesive subgraphs discovery. 

\begin{definition}
\textsf{\underline{Link-space graph}~\cite{lim2014linkscan}}.
Given a graph $G$, its corresponding link-space graph $LS(G)$ is defined as follows:
  \begin{itemize}
    \item A node $v_{i,j}$ in $LS(G)$ represents the link $e_{i,j}=\{i,j\}$ in $G$
    \item Two nodes $v_{i,k}$ and $v_{j,k}$ in $LS(G)$ are adjacent if and only if their corresponding links share a common node in $G$
    \item The weight $w(v_{i,k},v_{j,k})$ for the link $\{v_{i,k},v_{j,k}\}$ in $LS(G)$ is assigned by similarity $\sigma(e_{i,k},e_{j,k})$ calculated on $G$.
  \end{itemize}
\end{definition}

The link-space graph~\cite{lim2014linkscan} was proposed for identifying an overlapping community structure in a graph. Given a link-space graph $LS(G)$, the weight of a link $\{v_{i,k},v_{j,k}\}$ is defined as 
$w(v_{i,k},v_{j,k})=\sigma(e_{i,k},e_{j,k})=\frac{|\Gamma(i) \cap \Gamma(j)|}{|\Gamma(i) \cup \Gamma(j)|}$, where $\Gamma(i)=\{i \cup N(i)\}$.
It is a similarity between two incident links calculated on $G$ by measuring the Jaccard-type similarity between two different end nodes. The link-space graph has several benefits: 
(1) it helps us understand the structure of a graph with the language of links in order to capture high-order relationships; 
(2) it helps reveal overlapping community structures efficiently. However, even if the link-space graph is useful, it has several limitations (See the below example). 

\begin{example}\label{example:lim}
Here we introduce two examples to show the limitations of the link-space graph from the perspective of efficiency and effectiveness. 
  \begin{itemize}
    \item Efficiency : the link-space graph is not efficient as it generates a high number of edges. For example, suppose that there is a node $v$ with $1,000$ neighbor nodes. Then, its corresponding link-space graph contains a clique containing $1,000$ nodes. In addition, we observe that most edges in the link-space graph have small weights, i.e., are meaningless edges. 
    \item Effectiveness : In a link-space graph, an identified cohesive subgraph may contain unrelated nodes named \textit{free-riders} as a result. 
  \end{itemize}
\end{example}

In section~\ref{sec:comparingLSLK}, we show the detailed benefits of the link-skein graph compared with the link-space graph. Note that the link-skein graph only keeps the relatively important structures in a graph based on the triangles by pruning several low-weight edges. 

In this paper, we newly propose the link-skein graph, which is based on the link-space graph~\cite{lim2014linkscan} with improved efficiency and effectiveness. The definition of a link-skein graph is as follows. 

\begin{definition}
\textsf{\underline{Link-skein graph}}. 
Given a graph $G$, its corresponding link-skein graph $L(G)$ is defined as follows:
  \begin{itemize}
    \item A node $v_{i,j}$ in $L(G)$ represents the link $e_{i,j}=\{i,j\}$ in $G$
    \item Two nodes $v_{i,k}$ and $v_{j,k}$ in $L(G)$ are adjacent if and only if their corresponding links are contained in a triangle in $G$.
    \item The weight $w(v_{i,k},v_{j,k})$ on the link $\{v_{i,k},v_{j,k}\}$ in $L(G)$ is assigned by similarity $\sigma(e_{i,k},e_{j,k})$ calculated on $G$.
  \end{itemize}
\end{definition}

\begin{figure}[t]
 \centering
 \includegraphics[width=0.99\linewidth]{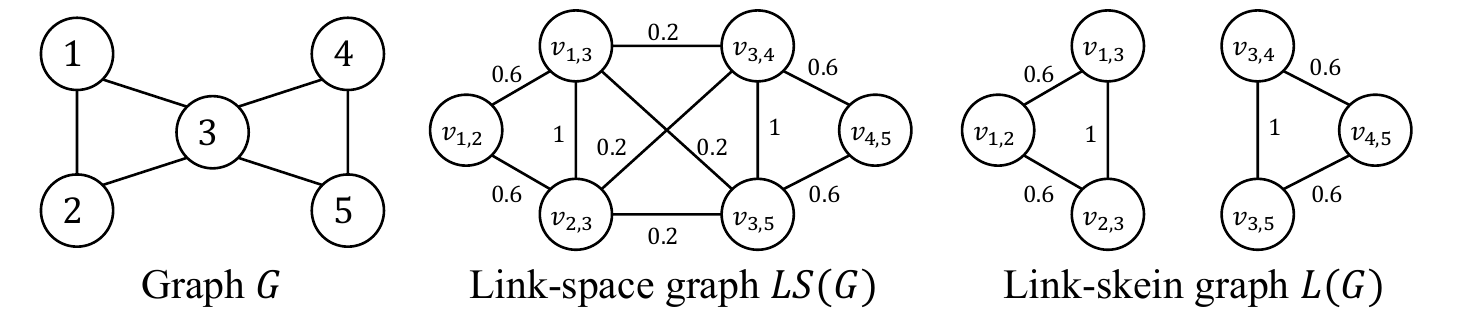}
 \caption{Graph, link-space graph, and link-skein graph }
 \label{fig:link-skein-gen}  
\end{figure}

We notice that a link-skein graph is a spanning subgraph of a link-space graph, or a graph sparsification due to the elimination of less important edges. Figure~\ref{fig:link-skein-gen} contrasts the link-skein graph with the link-space graph. The link-space graph has 10 edges, whereas the link-skein graph has only 6 edges with the same weights. The edges with the smallest weights in the link-space graph do not appear in the link-skein graph.

\subsection{Overlapping Cohesive Subgraph with Minimum Degree}
We first introduce some basic definitions before introducing our problem. 

\begin{definition}
\underline{$k$-core}~\cite{seidman1983network}. Given a graph $G=(V,E)$ and positive integer $k$, $k$-core of $G$, denoted by $D_k$ is a maximal subgraph consisting of a set of nodes of which all the nodes in $D_k$ have at least $k$ neighbor nodes.
\end{definition}


There are two important characteristics of the $k$-core: (1) \textit{Uniqueness}: given a graph $G$ and integer $k$, $k$-core is unique due to its maximality constraint; (2) \textit{Hierarchical structure} : $k$-core has a hierarchical structure, i.e., $(k+1)$-core $\subseteq$ $k$-core $\subseteq $ $(k-1)$-core when $k > 1$. As $k$-core satisfies the minimum degree constraint, we can use a set of connected components of $k$-core as the baseline of our algorithm. We next discuss our objective function.


\begin{definition}\label{def:linkdensitydef}
\textsf{link-density}. Given sets of nodes $\mathcal{C} = \{c_1, c_2,$ $\cdots, c_t\}$ where $\forall c\in \mathcal{C}, c\subseteq V_{L(G)}$, the link-density of $\mathcal{C}$ is defined as follows. 
\begin{displaymath}
\gamma (\mathcal{C}) = \sum_{c\in \mathcal{C}} \frac{\sum_{e_{i,j} \in E_{c}} w'(i,j)}{|V_{c}|},
\end{displaymath}
where $V_c$ is a set of nodes $c\subseteq V_{L(G)}$ in the link-skein graph and $E_c$ describes the edges in the link-skein graph, which is induced by a subgraph $c$, and $w'(i,j) = \frac{w(i,j)}{O(i,j)}$, where $O(i,j)$ indicates the number of appearances of $e_{i,j}$ in the subsets of $\mathcal{C}$.
\end{definition}

\begin{example}
In Figure~\ref{fig:link-skein-gen}, suppose that $t=1$ and we have two candidate subgraphs induced by the nodes $\{1,2,3\}$ (small cohesive subgraph) and $\{1,2,3,4,5\}$ (large cohesive subgraph).
We then check the link-densities of the two candidate subgraphs in the link-space and link-skein graphs, respectively. Figure~\ref{fig:linkdensitycanddiatesubgraph} reports the link-densities of the candidate subgraphs. We can notice that the link-space graph prefers large-sized subgraphs in terms of the link-density. 
However, our link-skein graph does not prefer large subgraphs. 
This helps us in identifying more densely connected overlapping cohesive subgraphs. 

\begin{figure}[ht]
 \centering
 \includegraphics[width=0.7\linewidth]{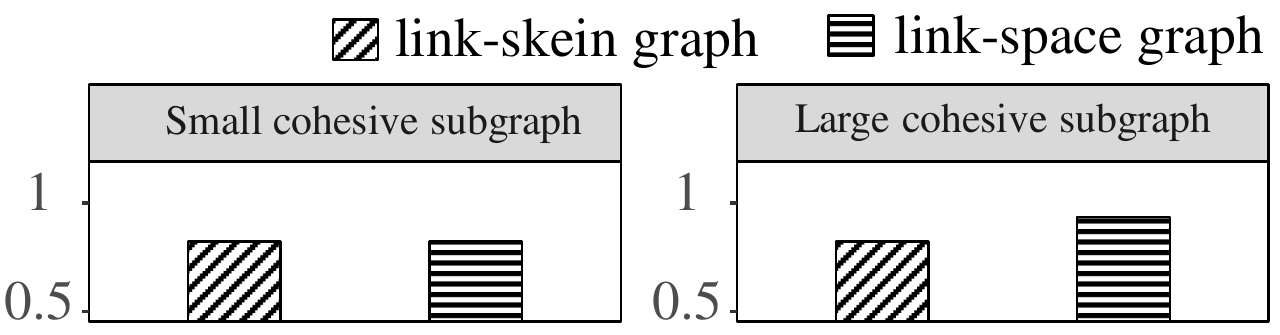}
 \caption{link-density of the candidate subgraphs}
 \label{fig:linkdensitycanddiatesubgraph}  
\end{figure}

Instead of directly using traditional graph density, we develop a new graph density measure for finding overlapping cohesive subgraphs. The rationale behind adding the occurrence term $O(i,j)$ is to prevent finding dense subgraphs, which commonly share nested dense subgraphs. For example, given a graph $G$ and $k=1$, suppose that there is a clique $C\subseteq V$. Then, finding top $t$ subgraphs is equal to finding the clique $C$ $t$ times as this can maximize the graph density. It indicates that the traditional graph density measure is not proper for finding overlapping dense subgraphs. To handle this problem, some studies adopted additional hyper-parameters, such as the overlapping ratio threshold~\cite{balalau2015finding} or distance~\cite{galbrun2016top}. 
In this paper, we do not use any thresholds or constraints to control the overlapping ratio. Instead of controlling the overlapping ratio, we use the link-density, which has the effect of lessening the edge weight if the edge has already been selected. When $t=1$, the link-density is the same as the traditional graph density. 

\begin{property}
Link-density is less sensitive than graph density when handling nested subgraphs.  
\end{property}

Let us suppose that we have two subgraphs $C_1$ and $C_2=C_1\cup \{u\}$ in a link-skein graph. When we calculate the graph density, it is $\frac{\sum_{e(u,v)\in E_{C_1}} w(u,v)}{|V_{C_1}|} + \frac{\sum_{e(u,v)\in E_{C_2}} w(u,v)}{|V_{C_2}|}$. If $C_1$ is densely connected, it can easily be that $C_2$ may have a high graph density since $C_2$ contains $C_1$. It indicates that when we would like to find $t$ subgraphs and $C_1$ is a clique, we may find $t$ subgraphs of which each contains $C_1$.

In section~\ref{sec:comparingLSLK}, we discuss the differences between link-space and link-skein graphs for small-sized datasets. 

\end{example}

\begin{figure}[ht]
 \centering
 \includegraphics[width=0.99\linewidth]{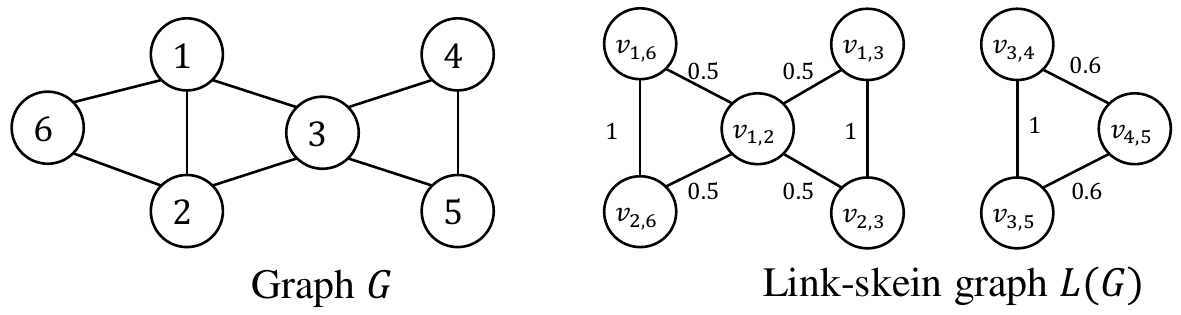}
 \caption{A toy network containing $6$ nodes}
 \label{fig:example2}  
\end{figure}

\begin{example}
We next present an example for computing the link-density when $t > 1$. Suppose that $t=2$ and we have two candidate solutions.
\begin{itemize}
  \item $S1 \Rightarrow \{v_{1,2},v_{1,6},v_{2,6}\}, \{v_{1,2},v_{1,3},v_{2,3}\}$
  \item $S2 \Rightarrow \{v_{1,2},v_{1,6},v_{2,6}\}, \{v_{1,2},v_{1,6},v_{2,6},v_{1,2},v_{1,3},v_{2,3}\}$
\end{itemize}
The link-densities of two solutions are as follows. 
\begin{align}
  \gamma(S1) &= \frac{0.5+0.5+1}{3} + \frac{0.5+0.5+1}{3} = 1.3333 \\
  \gamma(S2) &= \frac{0.25+0.25+0.5}{3} + \frac{0.25+0.25+0.5+0.5+0.5+1}{5} = 0.9333
\end{align}
We can notice that we prefer $S1$ to $S2$ without requiring any specific parameters. 
\end{example}

\begin{definition}
\underline{Minimum occurrence}. Given a subgraph $H$ of link-skein graph $L(G)$, the minimum occurrence of $H$, denoted $\beta(H)$, is the minimum number of node occurrences when the link-skein graph $H$ is translated back to the original graph $R(H)$.
\end{definition}

\begin{example}
In Figure~\ref{fig:link-skein-gen}, $\beta(\{v_{1,2}, v_{1,3}, v_{2,3}\})=2$ as nodes $1$, $2$, and $3$ appeared twice. $\beta(\{v_{1,2}, v_{1,3}, v_{2,3}, v_{3,4}\}])=1$ as node $4$ appeared once. 
\end{example}

Note that the minimum degree is closely related to the minimum occurrence. We notice that the following property always holds. 

\begin{property}
If a subgraph $H \in V_{L(G)}$ satisfies the minimum occurrence, it indicates that a subgraph $R(H)$ satisfies the minimum degree constraint.
\end{property}

\begin{proof}
The proof is trivial. 
\end{proof}

Given a link-skein graph $H$ of $L(G)$, $R(H)$ is a set of nodes in $G$ which are translated back from the link-skein graph $H$.
Now, we are ready to introduce our {\textsf{OCSM}} problem.

\begin{problemDefinition}
(Overlapping Cohesive Subgraphs with Minimum degree (\textsf{OCSM})). Given a graph $G=(V, E)$, a positive integer $k$, and the desired number of subgraphs $t$, {\textsf{OCSM}} aims for finding a set of subgraphs $H=\{H_1, H_2, \cdots, H_{t} \}$ where $\forall H_{i} \in H, H_i \subseteq V_{L(G)}$ such that
\begin{itemize}
  \item $\forall H_i\in H$, $R(H_i)$ is connected.
  \item $\forall H_i\in H$, $\delta(R(H_i))\geq k$. 
  \item $\gamma(H)$ is maximized.
\end{itemize}
\end{problemDefinition}

We call $\forall H_i\in H$, $\delta(R(H_i))\geq k$ as the degree constraint. Note that satisfying the degree constraint does not guarantee that the corresponding link-skein graph is connected. In Figure~\ref{fig:link-skein-gen}, the link-skein graph has two connected components, even if the original graph is connected. 

\begin{theorem}\label{theorem:the1}
The {\textsf{OCSM}} problem is NP-hard. 
\end{theorem}

\begin{proof}
We prove that our problem is NP-hard by reducing an instance of the {\kDS}~\cite{balalau2015finding} problem to our problem. {\kDS} problem is defined as follows: Given a graph $G$ and a positive integer $p$, and $\alpha\in [0,1]$, {\kDS} aims to find at most $p$ overlapping subsets $S_1, S_2, \cdots, S_p$ of the nodes, such that $\sum \rho (S_i)$ is maximized such that $\frac{|S_i \cap S_j|}{|S_i \cup S_j|} \leq \alpha, \forall S_i, S_j \in S$. They show that {\kDS} problem is NP-hard. 

Our reduction procedure is as follows. Suppose that we have an instance of {\kDS}: $I_{\kDS}=(G,p,\alpha=1)$. We can easily create an instance of our problem: $I_{\textsf{OCSM}}=(G', k=1, t=p)$ where $L(G')=G$. Then, finding the top $p$ densest subgraphs in {\kDS} is exactly the same as finding a solution of $I_{\textsf{OCSM}}$ since $\alpha$ = 1. Therefore, we can guarantee that our problem is NP-hard. 

\end{proof}


\begin{theorem}\label{theorem:the3}

Given an optimal solution $OPT$ of $L(G)$ when $t=1$, $\gamma(OPT) \leq w_{max}$ where $w_{max}$ is the maximum node weight in a link-skein graph. 
\end{theorem}

\begin{proof}
We notice the link-density of the optimal solution is less than or equal to $w_{max}$. 

\begin{align}
  \gamma(OPT) = \frac{\sum_{e(u,v)\in E_{OPT}} w(u,v)}{|V_{OPT}|} \leq w_{max} \\
  \sum_{e(u,v)\in E_{OPT}} w(u,v) \leq w_{max}|V_{OPT}| 
\end{align}

We can easily notice that $\gamma(OPT)\leq w_{max}$ holds since $w_{max}|V_{OPT}|$ is the maximum possible number of internal edges. It also indicates that given identified solution $C$, $\frac{\gamma(OPT)}{\gamma(C)}\leq \frac{w_{max}}{\gamma(C)}$, as $\gamma(C)$ is always positive. We notice that the approximation ratio depends on the link-density of our result.
\end{proof}

\begin{remark}
Note that if $k$-core of graph $G$ returns null, we fail to find a solution. Hence, computing the maximum coreness~\footnote{The coreness of a node is $k$ if it belongs to the $k$-core but not to $(k+1)$-core} may need to be a requirement for selecting a proper $k$ value in advance. 
\end{remark}

\subsection{Merits of the Link-Skein Graph}\label{sec:comparingLSLK}
\subsubsection{Efficiency}
To check the superiority of the link-skein graph as compared with the link-space graph, we use widely used networks, such as Football, Dolphin, Polbooks, and Karate~\cite{web:networkRepository}.
%
%
%
\begin{figure}[t]
 \centering
 \includegraphics[width=0.99\linewidth]{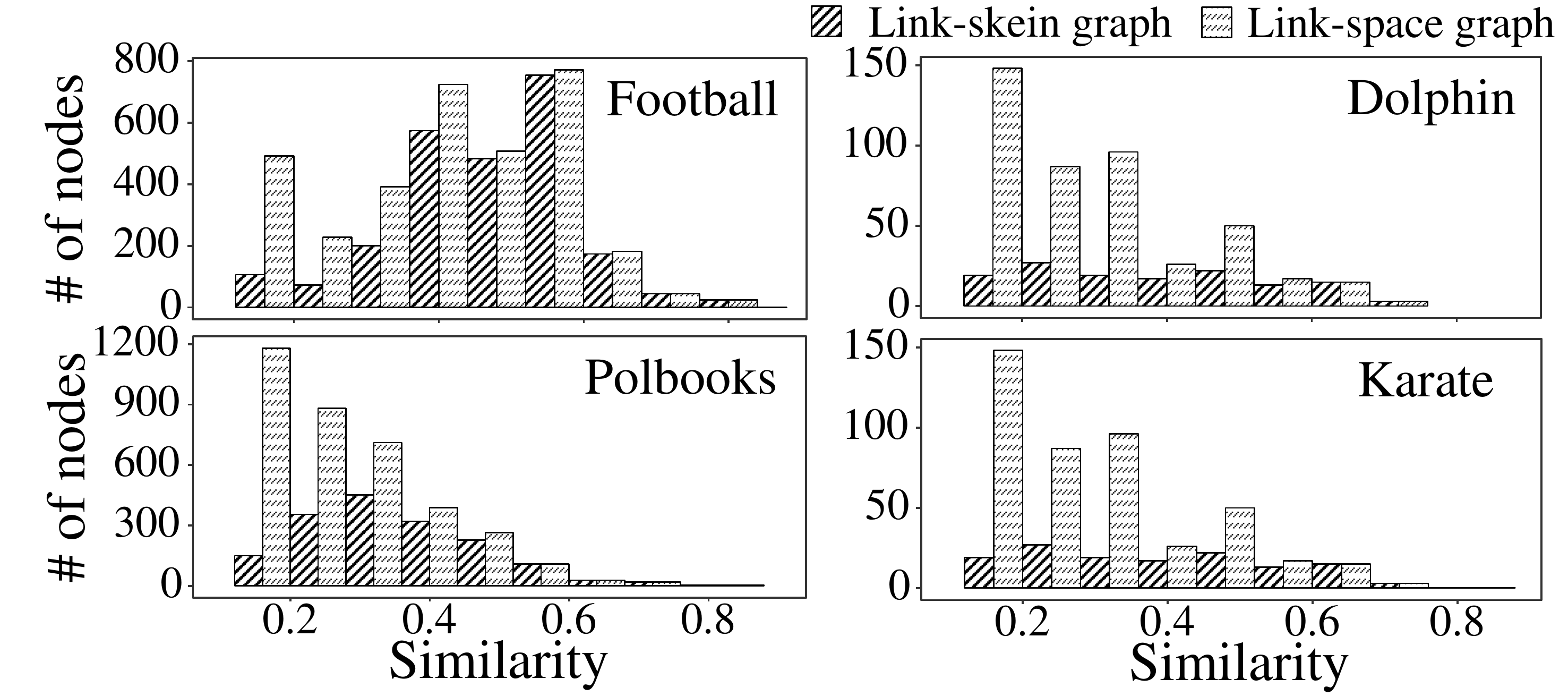}
 \caption{Difference of the similarity distribution between link-space and link-skein graphs}
 \label{fig:difference}
\end{figure}
In Figure~\ref{fig:difference}, we compare the similarity distributions of the link-skein and link-space graphs of the networks. From the experiments, we notice that from $0.5$ to $0.9$, the link-space and link-skein graphs show similar frequency trends. In particular, the larger the similarity value, such as $0.6$, the more similar are the observed frequency patterns.
However, when the similarity is very small such as 0.2, there are many edges in the link-space graph that rarely appear.
This is because the link-skein graph has the effect of pruning inessential edges which appear in the link-space graph. 
Therefore, we can consider that usage of the link-skein graph prunes relatively less important edges and keeps the important edges.

\subsubsection{Effectiveness}\label{sec:FRE}
We first introduce the free-rider effect problem~\cite{wu2015robust}. Let us denote $C_{OPT}$ as the optimal solution, whose goodness value $f(C_{OPT})\geq f(C), \forall C\subseteq V$ for the maximization problem. Note that there are two types of free-rider effects: (1) the global free-rider effect and; (2) the local free-rider effect. 
In this section, we do not consider the local free-rider effect since it is defined for a community search problem with query nodes~\cite{wu2015robust}. 
Thus, we only discuss the global free-rider effect. To check the global free-rider effect, we consider our problem as finding a single community (densely connected cohesive subgraph) without any query nodes. 

\begin{definition}\label{def:globalFRE}
\textsf{\underline{Global free-rider effect}~\cite{wu2015robust}}. 
A goodness function $f$ suffers from global free-rider effect if for any $C\subseteq V$, $f(C) \leq f(C\cup C_{OPT})$.
\end{definition}

It is known that many metrics, including minimum degree, graph density, modularity, and external conductance measures suffer from the free-rider effect~\cite{wu2015robust}. 
Our objective function (link-density) also suffers from the global free-rider effect, but we show that when we use the link-skein graph, we can mitigate the free-rider effect as opposed to when using the link-space graph. 
Let $f$ be the link-density in the link-space graph and $g$ be the link-density in the link-skein graph. 
 
\begin{theorem}\label{theorem:the2}
The link-skein graph mitigates the global free-rider effect compared with the link-space graph. 
\end{theorem}

\begin{proof}
Suppose that $C$ is a solution of $f$ and $g$, and $C_{OPT}$ as the optimal solution. From Definition~\ref{def:globalFRE}, we can derive the following inequalities.
\begin{align}\label{eq:FRE0}
  \forall C\subseteq V, f(C) \leq f(C\cup C_{OPT}) \\ \label{eq:FRE1}
  \forall C\subseteq V, g(C) \leq g(C\cup C_{OPT}) 
\end{align}

Let us denote $f(C)$ as the link-density of the link-space graph, $g(C)$ as the link-density of the link-skein graph, and $x_{C} = f(C) - g(C)$. 
Note that $x_C$ is always positive. This is because the denominator is the same, but the numerator of $f(C)$ is larger than $g(C)$. 
Similarly, let us denote $x_{C_{OPT}}$ as $f(C_{OPT}) - g(C_{OPT})$, and $x_{C,C_{OPT}}$ as $f(C\cup C_{OPT}) - g(C\cup C_{OPT})$. $x_{C_{OPT}}$ is the link-density gain due to additional edges in the link-space graph, and $x_{C,C_{OPT}}$ is the link-density gain due to additional edges between $C$ and $C_{OPT}$. 

We then check $f(C\cup C_{OPT})-f(C)$ for a comparison with $g(C\cup C_{OPT})-f(C)$. 

\begin{align}\label{eq:FRE2}
\begin{aligned}
  & f(C\cup C_{OPT})-f(C) \\
  \Leftrightarrow & g(C\cup C_{OPT}) +  x_{C, C_{OPT}} -(g(C) + x_{C})\\
  \Leftrightarrow & g(C\cup C_{OPT}) + (x_{C, C_{OPT}}-x_{C}) -g(C)
\end{aligned}
\end{align}

From Equations~\ref{eq:FRE1} and \ref{eq:FRE2}, we derive the following inequality since $x_{C, C_{OPT}}>x_{C}$. 
\begin{align}\label{eq:FRE3}
  f(C\cup C_{OPT}) - f(C) \geq g(C\cup C_{OPT}) - g(C)
\end{align}

Equation~\ref{eq:FRE3} implies that the link-space graph is more vulnerable with regards to the free-rider effects than the link-skein graph when we calculate link-density, as the link-space graph has additional edges with positive weights. 
\end{proof}

\section{Algorithms}
In this section, we introduce how we generate the link-skein graph and propose two algorithms to solve the {\textsf{OCSM}} problem. Each algorithm has different strategies to solve the problem: (1) {\textsf{APA}} is an advanced peeling algorithm which is a top-down approach by iteratively deleting a set of nodes based on the link-density contribution. It first focuses on the degree constraints as the major concern and then aims to maximize link-density; (2) {\textsf{SEA}} is a seed-based expansion algorithm which is a bottom-up approach by iteratively adding a set of nodes based on the criteria. Its main concern is for maximizing the link-density, then try to satisfy the degree constraints by expanding a set of nodes.

\subsection{Generating the link-skein graph}
Algorithm~\ref{alg:gen_link_skein} shows the procedure for generating the link-skein graph. 
It first calculates the similarity in the original graph to avoid duplicate computations, then assigns the weight in edges of the link-skein graph.

\spara{Time complexity.} Time complexity to generate the link-skein graph is $O(|E||V|)$ since $|E|$ is for calculating the similarity, and $|V|$ is to find common neighbor nodes. Due to the power-law distribution of the degree in a graph, the practical running time is reasonably faster than the theoretical time complexity.  


\setlength{\textfloatsep}{1pt}
\begin{algorithm}[ht]
\SetKw{break}{break}
\SetKw{return}{return}
\SetKw{next}{next}
\SetKw{true}{true}
\SetKw{AND}{AND}
\SetKw{null}{null}
\SetKwData{C}{C}
\SetKwData{H}{H}
\SetKwFunction{from}{from}
\SetKwFunction{to}{to}
\SetKwFunction{jaccard}{jaccard}
\SetKwFunction{intersect}{intersect}
\SetKwFunction{addNode}{addNode}
\SetKwFunction{get}{get}
\SetKwInOut{Input}{input}
\SetKwInOut{Output}{output}
\Input{$G=(V,E)$}
\Output{Link-skein graph $L(G)=(V_{L(G)}, E_{L(G)})$}
$\H \leftarrow \varnothing$ \;
\For{$e \in E$}{
  \H.\addNode{$e$}\;
}

\For{$e \in E$}{
  $u \leftarrow \from{e}$, $v \leftarrow \to{e}$\;
  $sim \leftarrow \jaccard(N(u) \cup u, N(v)\cup v )$\;
  $W \leftarrow \intersect{N(u), N(v)}$\;
  \For{$w \in W$}{
    $l_1 \leftarrow \H.\get{u,w}$\;
    $l_2 \leftarrow \H.\get{v,w}$\;
    $sim(l_1, l_2) \leftarrow sim$\;
  }
}
\return \H\;
\caption{Generating link-skein graph }
\label{alg:gen_link_skein}
\end{algorithm} 

\subsection{Advanced Peeling Algorithm (\textsf{APA}) }\label{sec:APA}

We first introduce the Peeling Algorithm ({\textsf{PA}}) which uses a straightforward approach. 
This algorithm is to use the $k$-core and minimum occurrence for finding a solution. Let denote a subgraph is \textit{feasible} when the subgraph in $G$ is connected and satisfies the minimum degree constraint, or a subgraph in $L(G)$ satisfies the minimum occurrence constraint.
In {\textsf{PA}}, $k$-core is used to find a maximal feasible solution in $G$. The high-level idea of {\textsf{PA}} is as follows. 
It firstly computes $k$-core $D_k$ to find feasible subgraphs in $G$. It then converts $D_k$ on the link-skein graph. Note that each subgraph in $D_k$ might be divided into multiple connected components in $L(D_k)$. For each connected component in $L(D_k)$, we check whether the subgraph is feasible, i.e., the subgraph satisfies the minimum occurrence. If the subgraph is not feasible, we iteratively delete a set of nodes whose occurrence is less than $k$ in a cascading manner. Finally, we pick the top $t$ subgraphs as a result. We notice that {\textsf{PA}} may return large-sized subgraphs as a result since the peeling procedure is a kind of finding maximal feasible subgraphs in the link-skein graph.

\begin{algorithm}[ht]
\SetKw{break}{break}
\SetKw{return}{return}
\SetKw{next}{next}
\SetKw{true}{true}
\SetKw{AND}{AND}
\SetKw{null}{null}
\SetKwData{C}{C}
\SetKwData{max}{max}
\SetKwData{ks}{$k^*$}
\SetKwData{kp}{$k'$}
\SetKwData{Dp}{$D_{k'}$}
\SetKwData{T}{T}
\SetKwData{C}{C}
\SetKwFunction{check}{checkParam}
\SetKwFunction{maxCoreIndex}{maxCoreIndex}
\SetKwFunction{CC}{connectedComps}
\SetKwFunction{add}{add}
\SetKwFunction{densestSubgraph}{densestSubgraph}
\SetKwFunction{kcore}{$k$-core}
\SetKwFunction{kpcore}{$k'$-core}
\SetKwFunction{top}{top}
\SetKwFunction{setupCoreIndex}{setupCoreIndex}
\SetKwFunction{remove}{remove}
\SetKwFunction{PAFunction}{PA}
\SetKwFunction{R}{R}
\SetKwFunction{best}{smallestAvgEdgeWeight}
\SetKwFunction{notSatisfying}{notSatisfying}
\SetKwFunction{L}{L}
\SetKwFunction{pickBest}{pickBest}
\SetKwInOut{Input}{input}
\SetKwInOut{Output}{output}
\Input{$G=(V,E)$, $k$, and $t$ }
\Output{{\textsf{OCSM}} $C\subseteq V$}
\C $\leftarrow \varnothing$\;
\While{$|\C| \neq t$}{
  $i \leftarrow 1$\;
  $T_i$ $\leftarrow$ \L{\PAFunction{$G, 1$}}\;
  \While{$|T_i| = 0$}{
    $v \leftarrow$ \best{$T_i$}\;
    $T_{i+1} \leftarrow T_i \setminus v$\;
    $T_{i+1} \leftarrow T_{i+1} \setminus$ \notSatisfying{$T_{i+1}$, $k$}\;
    $i \leftarrow i+1$\;
  }
  $T^* \leftarrow $ \pickBest{$T_1, T_2, \cdots, T_{i-1}$}\;
  \C $\leftarrow$ \C $\cup$ \R{$T^*$}\;
  Change the edge weight in \L{$G$}\;
}
\return \C\; 
\caption{Advanced Peeling algorithm(\textsf{APA})}
\label{alg:APA}
\end{algorithm} 
 
To overcome the limitation of {\textsf{PA}}, we propose an advanced peeling algorithm (\textsf{APA}) by considering the link-density in the peeling procedure of {\textsf{PA}}.
The procedure of the {\textsf{APA}} is described in Algorithm~\ref{alg:APA}. 
\begin{enumerate}
\item We firstly pick a feasible solution $T_1$ from {\textsf{PA}} having the largest link-density (Line 4);
\item Next, for the selected connected component $T_1$, we apply a density-based peeling strategy. We first pick a node $v$ having the smallest average edge weight then delete it. Next, we apply an occurrence-based peeling approach to guarantee the minimum degree constraint. If the link-density is improved, we keep the result. This process is repeated until the connected component becomes empty (Lines 5-10);
\item Among the intermediate subgraphs $T_1, T_2, \cdots, T_{i-1}$, we pick a subgraph which has the largest link-density when it is added to the current solution and adds it to the current solution (Lines 11-12). Next, we change the edge weight of the link-skein graph $G$ (Line 13);
\item Repeat steps 1 through 3 until finding the top $t$ subgraphs. 
\end{enumerate}

\begin{example}
\begin{figure}[ht]
 \centering
 \includegraphics[width=0.99\linewidth]{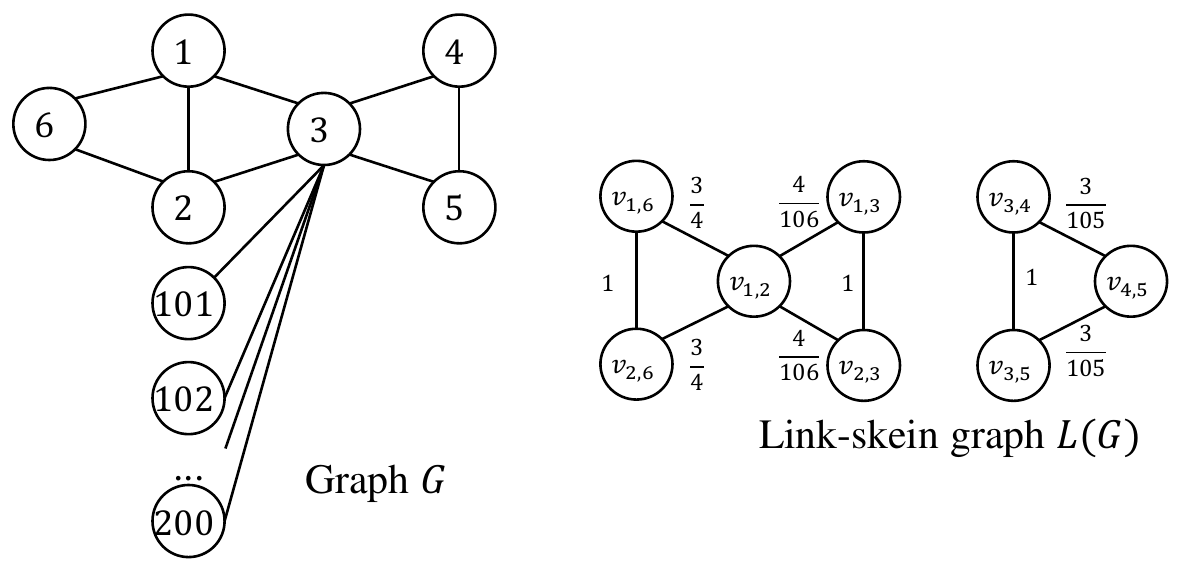}
 \caption{A toy network containing several nodes having small degree}
 \label{fig:example3} 
\end{figure}
We utilize Figure~\ref{fig:example3} to explain the procedure of {\textsf{APA}}. Suppose that $t=1$ and $k=2$. {\textsf{APA}} firstly finds a solution of {\textsf{PA}}. There are two connected components in the link-skein graph. We choose the larger one since its link-density is larger than the smaller one. Next, for every node, we compute the average node weight. For example, the node weight of $v_{1,3}$ is $0.519$. Since the node $v_{1,3}$ has the smallest node weight, we remove it. We then notice that the node $3$ does not satisfy the minimum degree constraint. Thus, the node $v_{2,3}$ is deleted together. This process is repeated until there is no node in the current subgraph. Finally, we return a subgraph $v_{1,6}, v_{1,2}, v_{2,6}$ as a result since its link-density is larger than other intermediate subgraphs. Since $t=1$, we do not need to update the edge weight. Whenever $t\geq 2$, it is required to update the edge weight based to avoid finding nested subgraphs. 
\end{example}

\spara{Limitation.} One issue is that after removing a node based on the average edge weights, a set of nodes can be deleted together cascadingly since the occurrence of some nodes can be decreased. This set of nodes is changed dynamically when we remove any node. Ideally, for every node, we can compute a set of nodes to be deleted together, then delete them which have the smallest link-density. However, this approach is prohibitive since computing all the sets in each iteration takes $O(|V||E|)$, and it cannot be utilized to handle a large-scale dataset. Thus, in {\textsf{APA}}, we designed that the node deletion is done independently to improve the running time even if we might lose additional accuracy.

\spara{Time complexity.} Time complexity of {\textsf{APA}} is as follows. 
\begin{itemize}
  \item $O(|V|+|E|)$ to get $k$-core and a set of connected components
  \item $O(|E_{L(G)}|)$ to apply the peeling approach for each iteration
  \item $O(|V_{L(G)}|)$ is the maximum number of iterations
  \item $O(|V||E|)$ is to compute the link-skein graph (See Algorithm~\ref{alg:gen_link_skein})
\end{itemize}
Therefore, the time complexity of {\textsf{APA}} is $O(t|E_{L(G)}||V_{L(G)}|+|V||E|)$ since normally $|E_{L(G)}||V_{L(G)}| >> |V|+|E|$. 
Note that the time complexity of {\textsf{APA}} is the same as {\textsf{PA}} since the additional peeling step takes the same computational cost of the peeling approach in {\textsf{APA}}. Note that it does not take much time normally to apply the peeling approach.

\subsection{Seed-Based Expansion Algorithm (\textsf{SEA}) }\label{sec:SEA}

In this section, we introduce the Seed-based Expansion algorithm (\textsf{SEA}) which is a bottom-up manner. {\textsf{SEA}} algorithm uses expansion approaches by combining Goldberg's densest subgraph algorithm~\cite{goldberg1984finding} and a local expansion approach~\cite{cui2014local} with a reweighting scheme. 
Instead of finding a solution by iteratively removing a set of nodes, this algorithm aims to find the densest subgraph and then iteratively expand the solution while satisfying the constraint with two criteria. 
There are three main operations: (1) finding the densest subgraph in $L(G)$; (2) applying local expansion; (3) reweighting; These operations are applied iteratively until finding top $t$ subgraphs. The detailed explanation of each operation is as follows. 

\spara{Goldberg's densest subgraph.} Goldberg~\cite{goldberg1984finding} proposes a polynomial time algorithm to find the densest subgraph by using the max flow. Goldberg's algorithm iteratively computes the minimum $s-t$ cut based on the binary search procedure. One limitation of the Goldberg's algorithm is that it can fail to find a solution in a large-scale dataset due to its computational cost. In our problem, we use Goldberg's algorithm in the link-skein graph to find seed nodes. 

\spara{Local expansion.} In \cite{cui2014local}, authors propose two greedy strategies to find a community satisfying the minimum degree constraints from a seed node : (1) 
largest increment of goodness (lg). This approach is to choose a node having the largest $\delta(G[C \cup v]) - \delta(G[C])$ in the expansion stage; (2) largest number of incidence (li). It chooses the node with the largest number of connections to the current node in the expansion stage, i.e., $f(v) = deg_{G[C\cup v]} (v)$. We use both strategies in the local expansion process. Note that our operation is in the link-skein graph. Therefore, we use $\beta(L(G)[C \cup v]) - \beta(L(G)[C])$ for lg and $f(v) = deg_{L(G)[C\cup v]} (v)$ for li. 

\spara{reweighting scheme.} Since we aim to find the top $t$ subgraphs, it is required to have additional operations. Suppose that we have identified the top $1$ subgraph. The simple way is just to remove the subgraph in the link-skein graph, then find other subgraphs. 
However, this approach has a flaw. 
Let assume that there are two cliques $C_1$ and $C_2$ which are overlapped partially, i.e., half nodes of each clique are overlapped. 
Suppose that we have identified $C_1$ as the top $1$ subgraph and have removed it. Then $C_2$ may not be considered since some nodes in $C_2$ are already removed. Thus, we change the edge weight of the selected subgraphs in the solution to $0$ in the link-skein graph. It makes the Goldberg's algorithm return meaningful results to find the top $t$ subgraphs. 


\setlength{\textfloatsep}{1pt}
\begin{algorithm}[ht]
\SetKw{break}{break}
\SetKw{return}{return}
\SetKw{next}{next}
\SetKw{true}{true}
\SetKw{null}{null}
\SetKwData{C}{C}
\SetKwData{max}{max}
\SetKwData{T}{T}
\SetKwData{C}{C}
\SetKwFunction{reweighting}{reweighting}
\SetKwFunction{kcore}{$k$-core}
\SetKwFunction{goldberg}{goldberg}
\SetKwFunction{add}{add}
\SetKwFunction{expansion}{expansion}
\SetKwFunction{L}{L}
\SetKwInOut{Input}{input}
\SetKwInOut{Output}{output}
\Input{$G=(V,E)$, $k$, and $t$ }
\Output{{\textsf{OCSM}} $C\subseteq V_{L(G)}$}
\C, \T $\leftarrow \varnothing$\;
\T $\leftarrow$ \kcore{$G$}\;
$LT$ $\leftarrow$ \L{\T}\;
\While{$|\C| \neq t$}{
  $S \leftarrow $ \goldberg{$LT$}\;
  $S \leftarrow $ \expansion{$S$, $LT$}\;
  \If{$S = \varnothing$}{
    \next\;
  }
  \C.\add{$S$}\;
  \reweighting{$S$, $LT$}\;
}
\return \C\; 
\caption{Seed-based Expansion Algorithm(\textsf{SEA})}
\label{alg:SEA}
\end{algorithm} 

The pseudo description of {\textsf{SEA}} is described in Algorithm~\ref{alg:SEA}. Initially, we compute $k$-core and then convert the result of $k$-core to the subgraph of the link-skein graph (lines 2-3). Until finding the top $t$ subgraphs, we firstly find the densest subgraph which can be the seed nodes (line 5). Next, we use the local expansion manner~\cite{cui2014local} to expand the seed nodes to guarantee the degree constraint (line 6). If we identify a subgraph satisfying the degree constraint, we add it to the solution and change the weight of the link-skein graph (lines 9-10). Finally, we return the resulted subgraph as a result. 

\begin{example}
We reuse Figure~\ref{fig:example3} to explain the procedure of {\textsf{SEA}}. Suppose that $t=1$ and $k=2$. It firstly applies Goldberg's densest subgraph algorithm to find an initial subgraph. When we apply the algorithm, it returns $\{v_{1,6}, v_{1,2}, v_{2,6}\}$. Luckily, all the nodes satisfy the minimum degree constraint. Thus, we return the result directly. Otherwise, we iteratively add a set of nodes to satisfy the minimum degree constraint by applying \textsf{li} or \textsf{lg} methods. 
\end{example}

\spara{Time complexity.} Time complexity of {\textsf{SEA}} is as follows. 
\begin{itemize}
  \item $O(|V_{L(G)}|^3)$ to compute Goldberg's densest subgraph~\cite{goldberg1984finding}.
  \item $O(|V||E|)$ is to compute the link-skein graph (See Algorithm~\ref{alg:gen_link_skein})
  \item $O(X^*)$ as the time complexity for local expansion. lg takes $O(|V_{L(G)}|+|E_{L(G)}| \log{|V_{L(G)}|})$ and li takes $O(|V_{L(G)}|+|E_{L(G)}|)$
\end{itemize}
Therefore, the time complexity of {\textsf{SEA}} is $O(t (|V_{L(G)}|^3 +X^*+|V||E|))$.

\section{Experiments}
\label{sec:experiments}

In this section, we evaluate the proposed algorithms using real-world datasets. All experiments were conducted on Ubuntu 14.04 with a 32GB memory and 2.50GHz Xeon CPU E5-4627 v4. We used JgraphT library~\cite{jgrapht} in our implementation.

\subsection{Experimental Setup}

\spara{Dataset.} Table~\ref{tab:dataset} shows the statistics of 6 datasets in our experiments. All datasets are publicly available. We denote $CI$ as the maximum core index, $AD$ as the average degree, and \# $\bigtriangleup$ as the number of triangles.

\begin{table}[t]
\caption{Summary of the real-world datasets}
\label{tab:dataset}
\centering
\begin{tabular}{l|c|c|c|c|c}
\hline
Name    & \textbf{\# nodes}   & \textbf{\# edges}  & \textbf{$CI$} & \textbf{$AD$} & \textbf{\# $\bigtriangleup$}     \\ \hline \hline
Amazon\cite{yang2015defining}   & 334,863  & 925,872  & 6 & 5.52 & 667,129 \\ \hline
Brightkite\cite{cho2011friendship} & 58,228  & 214,078  & 52 & 7.35 & 494,728   \\ \hline
DBLP\cite{yang2015defining}   & 317,080  & 1,049,866 & 113 & 6.62 & 2,224,385  \\ \hline
Hepth\cite{leskovec2007graph}  & 9,877   & 25,998  & 31 & 5.26 & 28,339    \\ \hline
LA\cite{bao1, bao2}       & 500,597  & 1,462,501 & 120 & 5.84 & 710,243  \\ \hline
Youtube\cite{yang2015defining} & 1,134,890 & 2,987,624 & 51 & 5.27 & 3,056,386  \\ 
\hline \hline
\end{tabular}
\end{table}

\spara{\bf Algorithms.} To the best of our knowledge, our {\textsf{OCSM}} does not have direct competitors in previous literature due to the overlapping and minimum degree constraints. Thus, we compare the proposed algorithms with the several cohesive subgraph discovery problems including $k$-core, $k$-peak, $k$-truss, and $(3,4)$-nucleus in our experiments. As we aim to find the top $t$ subgraphs, we use a greedy manner for post-processing. The list of the algorithms is as follows. 

\begin{itemize}\setlength\itemsep{-0.1em}
  \item $k$-core~\cite{seidman1983network}\footnote{\url{https://igraph.org/}}
  \item $k$-peak~\cite{kpeak}\footnote{\url{https://github.com/priyagovindan/kpeak}}
  \item $(k+1)$-truss~\cite{cohen2008trusses}\footnote{\url{https://rdrr.io/github/alexperrone/truss/}}
  \item $(3,4)$-nucleus~\cite{sariyuce2015finding}\footnote{\url{https://github.com/sariyuce/nucleus}}
  \item {\textsf{PA}}: Peeling Algorithm (in Algorithm~\ref{sec:APA})
  \item {\textsf{APA}}: Advanced Peeling Algorithm (in Algorithm~\ref{alg:APA})
  \item {\textsf{SEA}}: Seed-based Expansion Algorithm (in Algorithm~\ref{alg:SEA})
\end{itemize}

\spara{\bf Parameter setting.} We use a different $k$ value based on the maximum core index. When $CI$ is less than $50$, we vary $k$ between $3$ and $6$. To test the effect of $t$, we fix $k=3$. 
When $CI$ is larger than $50$, we vary $k$ between $5$ and $8$ and set $k=5$ to test the effect of $t$. 
For selecting a proper $k$, we follow previous studies, which used the minimum degree threshold~\cite{fang2020survey,kim2020densely}.
Finally, the link-density is chosen to measure the quality of the output subgraphs while the running time is used to measure the efficiency of our algorithms.

\subsection{Experimental Results} \label{sec:experimental_result}


\begin{figure}[ht]
 \centering
 \includegraphics[width=0.99\linewidth]{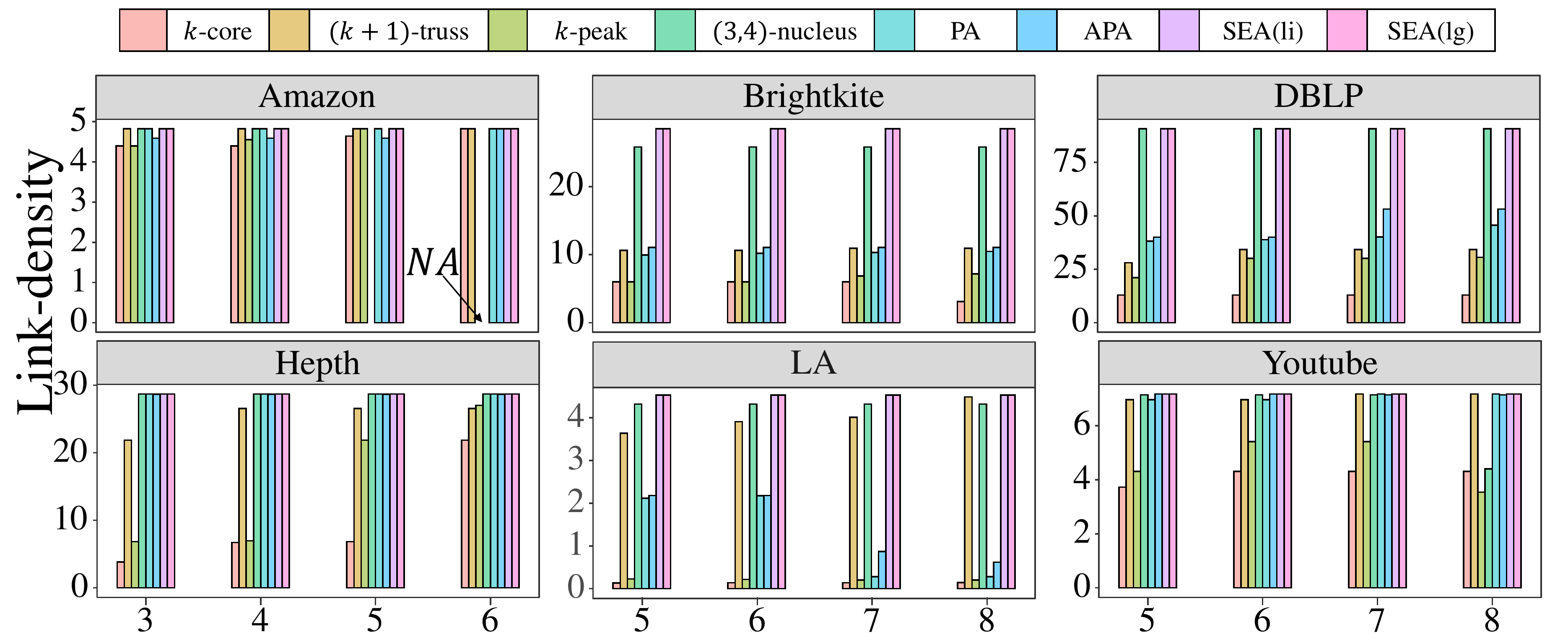}
 \vspace{-0.2cm}
 \caption{Effect of $k$}
 \label{fig:kVar}
\end{figure}

\spara{Effect of $k$.} Figure~\ref{fig:kVar} shows the resultant link-density when we change the value of $k$. We observe that our {\textsf{SEA}} algorithm has the largest link-density for all cases. When the $k$ value increases, the gap between our algorithms and other algorithms decreases. This is because the input graph size becomes small. Thus, there is not sufficient room for improving the link-density. For an Amazon dataset, $(3,4)$-nucleus does not return a subgraph with $k\geq6$. Thus, the $(3,4)$-nucleus does not return a solution for $k=6$. We can notice that the $(3,4)$-nucleus achieves comparable results 
for some datasets. However, we point out that the $(3,4)$-nucleus cannot return a result satisfying our desired $k$ minimum degree.

\begin{figure}[htb]
 \centering
 \includegraphics[width=0.95\linewidth]{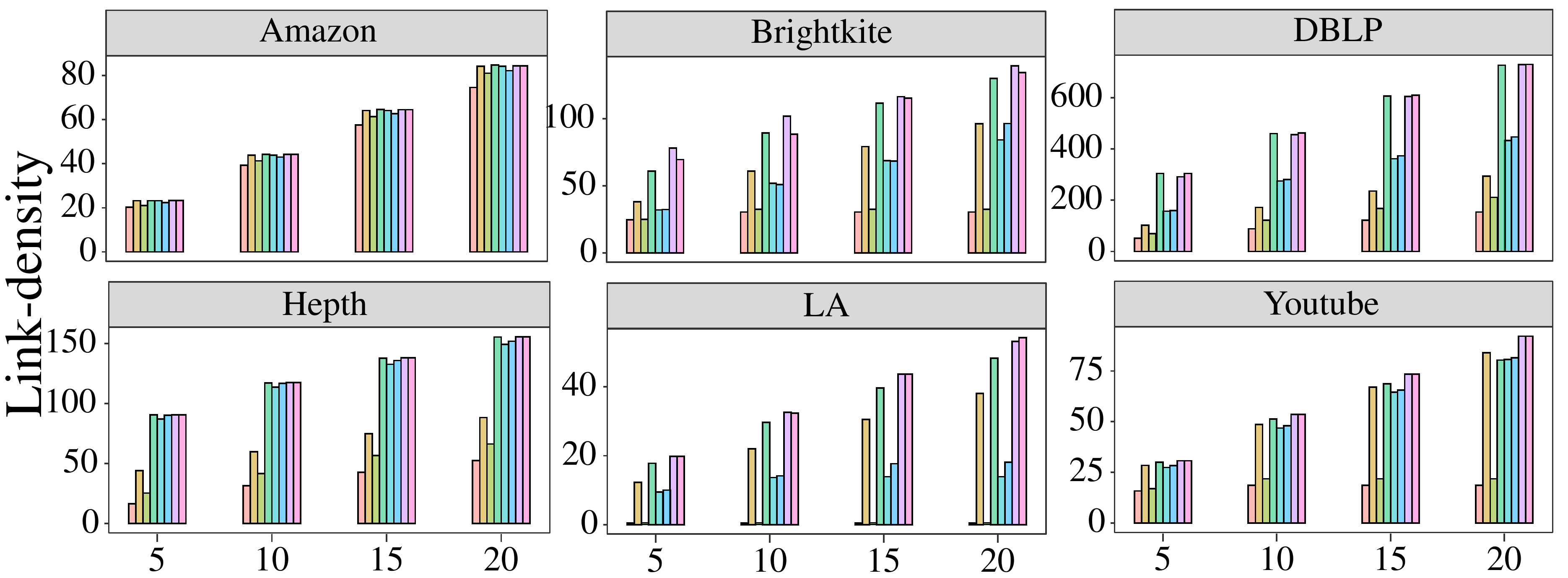}
 \vspace{-0.2cm}
 \caption{Effect of $t$ (use the same legend with Figure~\ref{fig:kVar})}
 \label{fig:nVar}
\end{figure}

\spara{Effect of $t$.} Figure~\ref{fig:nVar} shows the resulting link-density when we change the value of $t$. 
From among the three algorithms, the proposed {\textsf{SEA}} algorithm is considered as the best algorithm. When $t$ becomes large, we observe that the link-density increases and that the gap between ours and the existing algorithms becomes larger.

\begin{figure}[ht]
 \centering
 \includegraphics[width=0.99\linewidth]{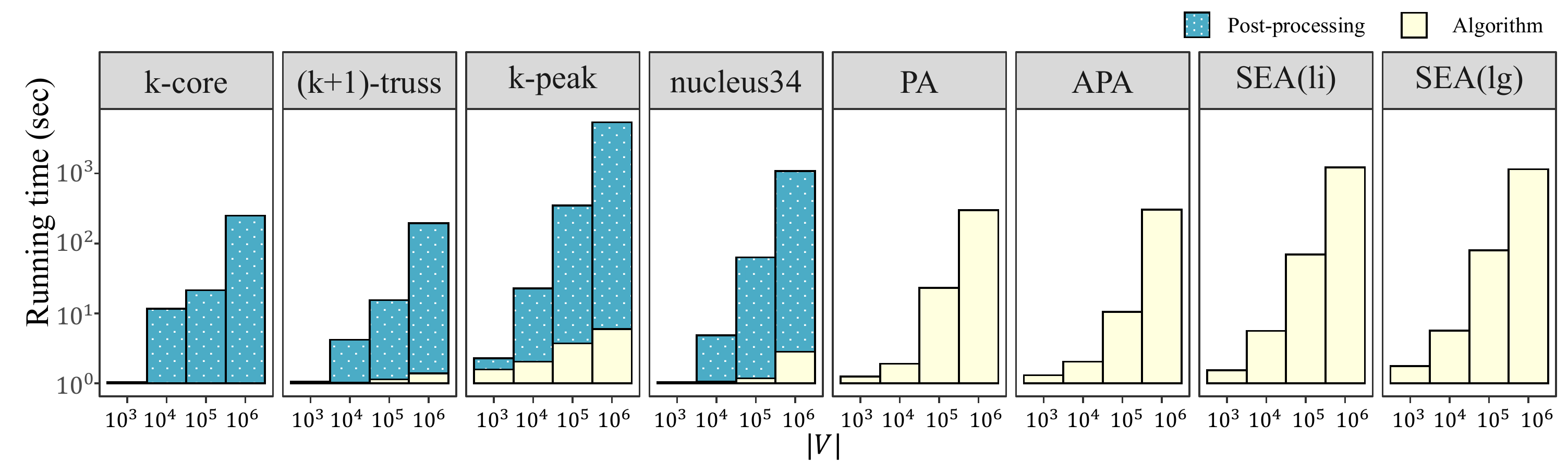}
 \caption{Scalability test}
 \label{fig:syn_scalability}  
\end{figure}

\spara{Scalability test.} To demonstrate the scalability of our algorithms, we vary the number of nodes between $1K$ and $1M$ in the LFR benchmark network~\cite{lancichinetti2008benchmark} using default parameters. For the scalability test, we fixed the number of the subgraph parameter $t$ to $20$ and the minimum degree threshold $k$ to $3$. Figure~\ref{fig:syn_scalability} shows the log-scaled running time of our algorithms and existing algorithms. 
For existing algorithms including $k$-core, $k$-truss, $k$-peak, and $(3,4)$-nucleus, there are two parts : algorithm running time and post-processing time. 
Algorithm running time is the time necessary to compute cohesive subgraphs and post-processing time is the time necessary to find top $t$ subgraphs based on link-density in a greedy manner. 
We observe that the running times of the proposed {\textsf{PA}} and {\textsf{APA}} are not significantly different and are comparable with that of $k$-core, which is the fastest algorithm from among the existing algorithms. However, {\textsf{SEA}} takes much longer since it needs to find the densest subgraph to find seed nodes. However, we notice that our {\textsf{SEA}} is more scalable than $k$-peak and nucleus decomposition, as our algorithms do not require the enumeration of all the possible solutions.
In $k$-peak or nucleus, a huge number of subgraphs are enumerated as a result. Thus, it naturally takes a long time to pick top $t$ subgraphs. 

\begin{table}[ht]
\caption{Effectiveness test}
\label{tab:karate_exact}
\centering
\begin{tabular}{c|c|c|c|c}
\hline
Algorithm       & Exact solution & $k$-core & ($k$+1)-truss & $k$-peak    \\ \hline
\textit{Link-density} & \textcolor{red}{3.21}      & 0.75  & 1.12    & 0.701     \\ \hline \hline
Algorithm       & \textsf{PA}       & \textsf{APA}  & {\textsf{SEA}(li)}   & \textsf{SEA}(lg)    \\ \hline
\textit{Link-density} & 1.67      & 2.37  & 2.56    & \textbf{2.91} \\ \hline
\end{tabular}
\end{table}

\spara{Effectiveness.} We use the Karate network~\cite{zachary1977information}, which contains $34$ nodes and $78$ edges to check the effectiveness of our algorithm. We first enumerate all the connected subgraphs~\cite{alokshiya2018linear}, then filter them out if the minimum degree of a subgraph is smaller than or equal to $3$. In total, there are $3,431$ connected subgraphs satisfying the minimum degree constraint. Next, we use a brute-force approach to find the top $4$ subgraphs. 
In the experiment, we do not include the $(3,4)$-nucleus as it cannot be guaranteed to return any subgraph which has at least $3$ minimum degree. Table~\ref{tab:karate_exact} 
reports the results of the link-density. We notice that {\textsf{SEA}} with lg has a similar result as the exact solution and that {\textsf{APA}} also returns comparable results compared with the exact solution. In contrast, all competitors have low link-density.

\begin{figure*}[ht]
 \centering
 \includegraphics[width=0.99\linewidth]{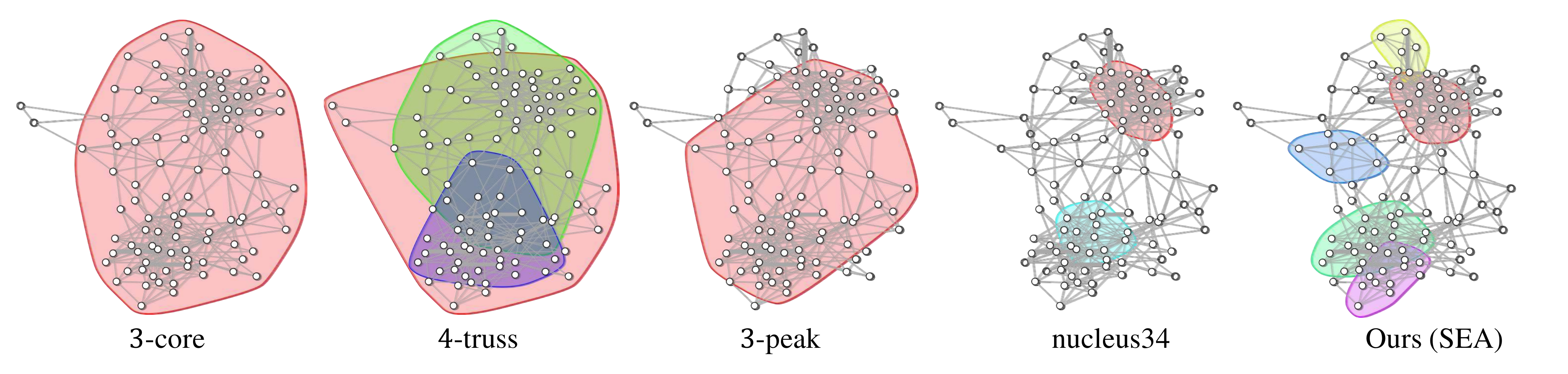}
 \caption{Resultant subgraphs discovered by the baseline algorithms and \textsf{SEA} on Polbooks dataset}
 \label{fig:casestudy}  
\end{figure*}

\spara{Case study (Polbooks).} Figure~\ref{fig:casestudy} shows a case study for the Polbooks network~\cite{web:networkRepository}.
We set the parameters $k=3$ and $t=5$. Note that our result is translated back to the original graph for visualization. We observe that $k$-core and $k$-peak return a single giant connected component, which is loosely connected. This phenomenon occurs more frequently when the value of $k$ is small. We also observe that $k$-truss returns three connected components which are not sufficiently cohesively connected. $(3,4)$-nucleus returns 2 connected components which are densely connected components. We notice that our {\textsf{SEA}} returns only $5$ clear cohesive connected components for which each subgraph satisfies the minimum degree constraint. Note that each connected component is densely connected and we can notice that several nodes overlap.

\begin{figure}[ht]
 \centering
 \includegraphics[width=0.99\linewidth]{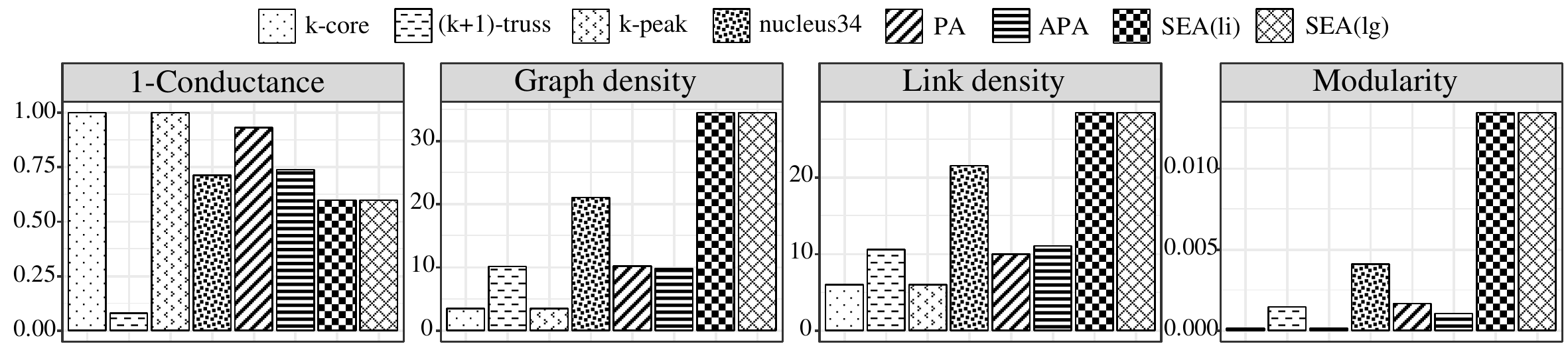}
 \caption{Comparing algorithms with different measures}
 \label{fig:tradDensity}  
\end{figure}

\begin{figure}[ht]
 \centering
 \includegraphics[width=0.99\linewidth]{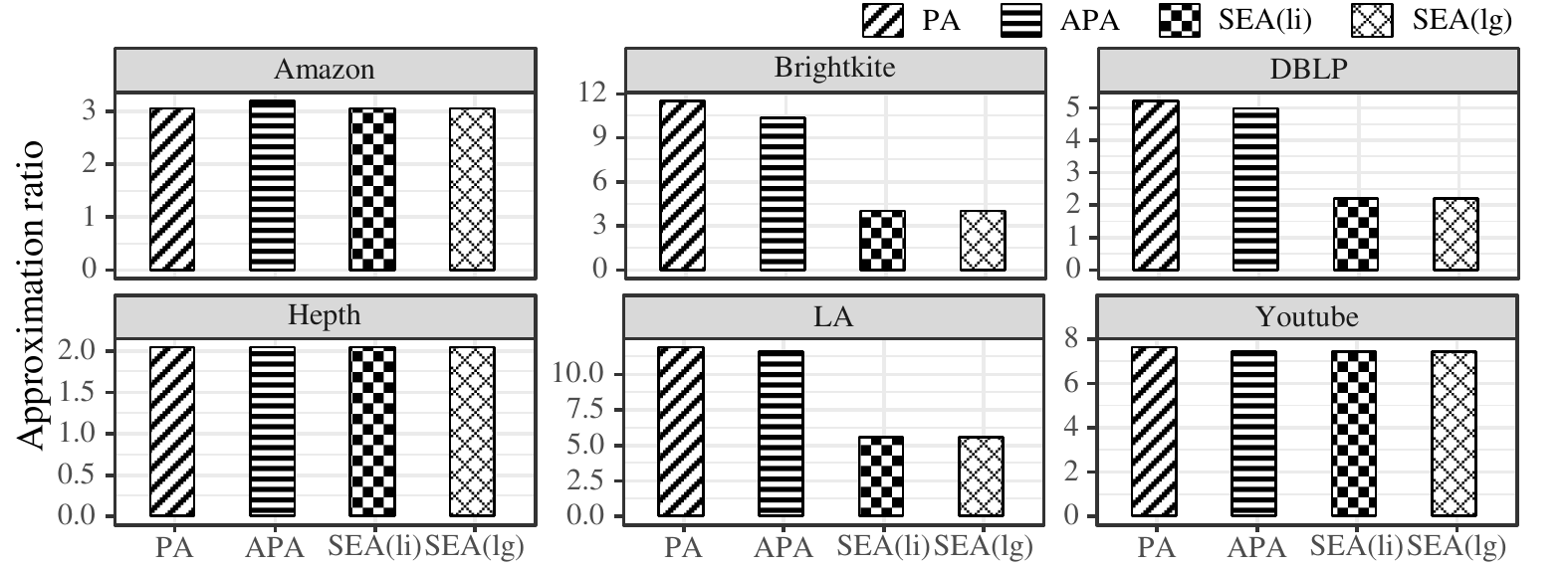}
 \caption{Approximation ratio of our algorithm}
 \label{fig:approx}  
\end{figure}

\spara{Comparing with different measures.} 
Figure~\ref{fig:tradDensity} reports on the four different measures of the resulting cohesive subgraphs for the Brightkite dataset. Note that a larger score indicates a better result for all measures. 
When we use the traditional graph density as an evaluation measure, we notice that the trend of the graph density is quite close to that of the link-density. However, as we have mentioned before, if our objective function is the traditional graph density, we cannot identify the overlapping structures. 
We also use graph modularity ~\cite{newman2006modularity} to measure the quality of the identified cohesive subgraphs. We notice that our \textsf{SEA} algorithm returns the largest modularity. 
For $1-$graph conductance, we observe that $k$-core and $k$-peak have large $1-$conductance as they return isolated connected components as a result. We also check that ($k+1$)-truss has the lowest $1-$conductance, which implies that the identified cohesive subgraph contains many external edges. In contrast, we use the link-skein graph, which helps consider external edges. Thus, our algorithms return high-quality dense cohesive subgraphs as a result.

\spara{Approximation ratio.} Figure~\ref{fig:approx} shows the result of the approximation ratio of the proposed algorithm for $k=5$ and $t=1$. We compute the approximation ratio from Theorem~\ref{theorem:the3}. We notice that the approximation ratio of the \textsf{SEA} algorithm is more reasonable than those of other algorithms. The range of the approximation ratio of \textsf{SEA} is between $2$ and $7$.

\section{Related work}\label{sec:relatedwork}

\subsection{ \textit{k}-core and Its Variations}\label{sec:kcore}
$k$-core is widely used to find cohesive subgraphs. The definition of the $k$-core~\cite{seidman1983network} is as follows: given a network $G=(V,E)$ and the positive integer $k$, the $k$-core of $G$, denoted by $H_k$, consists of a set of nodes of which all the nodes in $H_k$ have at least $k$ neighbor nodes in $H_k$. 
Batagelj et al.~\cite{BZ11, batagelj2003m} proposed an efficient $O(|E|)$ algorithm for finding the $k$-core. 
Sariyuce et al.~\cite{sariyuce2016incremental} studied an incremental $k$-core problem in a dynamic graph. 
Instead of finding the whole $k$-core for every insertion and deletion in a dynamic graph, they proposed efficient algorithms to avoid duplicate operations.

$k$-truss~\cite{cohen2008trusses} has recently been proposed for finding a cohesive subgraph. The definition is as follows: given a graph $G$ and the positive integer $k\geq2$, the $k$-truss of $G$ is a maximal subgraph in which all edges are contained in at least $(k-2)$ triangles within the subgraph. It is known that the time complexity of $k$-truss is $O(|E|^{1.5})$. Even if $k$-truss returns more cohesive subgraph, it is hard to find an appropriate parameter $k$. 

In ~\cite{kpeak}, the authors claimed that when the graph contains multiple distinct regions with different edge densities, $k$-core cannot handle the sparser regions. To handle this problem, they formulated $k$-peak decomposition problem, which aims to find the centers of distinct regions in the graph. They proposed an efficient $k$-peak decomposition algorithm with a rigorous theoretical analysis. 

There are several extensions of the $k$-core. Zhang et al.~\cite{zhang2017engagement} proposed $(k, r)$-core for an attributed social network. They denoted a connected subgraph $S \subseteq V$ where $S$ is a $(k,r)$-core if $S$ satisfies both structure constraint ($\delta(S)\geq k$) and the similarity constraint ($DP(S)=0$ where $DP(S)$ indicates that the number of dissimilar pairs in subgraph $S$). In the paper, they focused on two fundamental problems: enumerating all maximal $(k,r)$-cores and finding the maximum $(k,r)$-core. 
Recently, Bonchi et al.~\cite{bonchi2019distance} introduced $(k,h)$-core which considers the $k$-core with graph distance. Given a distance threshold $h$ and the positive integer $k$, $(k,h)$-core of $G$ is a maximal subgraph such that every node in $(k,h)$-core has at least $k$ $h$-neighborhoods. 

Bhawalkar et al.~\cite{bhawalkar2015preventing} propose the anchored $k$-core problem. The problem is to find a set of anchor nodes to maximize the size of the $k$-core. An anchor node indicates that a selected node, which does not belong to $k$-core, but is forced to belong to the $k$-core. 
They show that finding $b$ anchor nodes is NP-hard when $k\geq 3$. For a special case ($k=2$), they proposed the exact algorithm, which has polynomial time complexity. Zhang et al.~\cite{zhang2017olak} proposed the practical OLAK (onion layer based anchored $k$-core) algorithm to solve the anchored $k$-core problem efficiently by reducing the search space significantly. 

Sariyuce et al.~\cite{sariyuce2015finding} proposed the graph nucleus decomposition: given two positive integers $r < s$, the $k$-($r,s$)-nucleus is defined as a maximal union of
$s$-cliques, in which every $r$-clique is present in at least $k$ $s$-cliques, and any pair of $r$-cliques in that subgraph is connected via a sequence of $s$-cliques containing them. Thus, author mentioned that the $k$-($r,s$)-nucleus is
a generalized version of $k$-truss and $k$-core. In \cite{sariyuce2015finding}, they mentioned that 
when $r=1$ and $s=2$, the $k$-($1, 2$)-nucleus is a maximal subgraph with the minimum degree $k$, i.e., $k$-core. Similarly, when $r=2$ and $s=3$, $k$-($2, 3$)-nucleus is the same with the definition of $k$-truss. However, when the parameter $s$ becomes larger, the graph nucleus may suffer from a scalability issue, as it takes $\Omega(|E|^{\frac{s}{2}})$ where $|E|$ is the number of edges~\cite{huang2017attribute}.

\subsection{Finding the Densest Subgraph}

Finding the densest subgraph is one of the fundamental problems in the data mining field~\cite{wang2020finding}. Given a graph $G=(V,E)$, the goal is to find the subgraph of $G$ which has the highest density. One of the popular density metrics is defined as $\frac{|E|}{|V|}$. Goldberg~\cite{goldberg1984finding} proposed an exact algorithm by using the algorithms for the max-flow problem. It has polynomial time complexity, but it cannot address a large-size dataset. Charikar~\cite{charikar2000greedy} proposed an efficient top-down approach for finding the densest subgraph, which has the $2$-approximation ratio. The high-level idea is to iteratively remove the node with the minimum degree. For every iteration, it checks the density, then finally picks a subgraph with the maximum density. This algorithm is used when the graph size is large due to its efficiency. Tsourakakis~\cite{tsourakakis2015k} introduced the average triangle density and proposed exact and approximation algorithms. Balalau et al.~\cite{balalau2015finding} introduced a $(k,a)$-dense subgraph with the limited overlap ($(k,a)$-DSLO) problem. It finds at most $k$ subgraphs with the largest sum of subgraph density. Note that the overlapping ratio of any pair of subgraphs should be less than or equal to $a$. They showed that this problem is NP-hard and proposed heuristic algorithms. 
Galbrun et al.~\cite{galbrun2016top} also studied the top $k$ overlapping dense subgraph problem. They additionally considered the distance between subgraphs by adding a regularization parameter $\lambda$ to control for overlaps of the subgraphs. To solve the problem, they proposed a peeling algorithm, which holds a $\frac{1}{10}$-approximation ratio.

\subsection{Line Graph Analysis}
A node in real-world networks, especially in social networks, typically belongs to multiple communities, so communities overlap at a node. Overlapping community detection is known as a more generalized problem compared with disjoint community detection~\cite{lim2014linkscan}.
In order to tackle the overlapping community detection problem, an intuitive way is to identify a partition of links (i.e., relationships) rather than a partition of nodes (i.e., individuals).
Evans and Lambiotte~\cite{evans2009link} and Ahn et al.~\cite{ahn2010link} introduced the line graph model for this purpose.
The line graph of the original graph is constructed as follows: each node is
mapped from a link in the original graph that two nodes are adjacent if the corresponding links in
the original graph share a common node. 
The link-space graph~\cite{lim2014linkscan}, which is a variant of the line graph, is used for the state-of-the-art algorithms dealing with the overlapping community detection problem.
It is also a transformed graph of the original graph, whereby its topological structure is the line graph, and the weight is calculated on the original graph and carried over into the transformed graph.
The weighting scheme helps avoid unnecessarily large communities containing weak ties.

\section{Conclusion}
\label{sec:conclusion}
In this paper, we propose \underline{O}verlapping \underline{C}ohesive \underline{S}ubgraphs with \underline{M}inimum degree problem, which is called {\textsf{OCSM}}. We proved that the {\textsf{OCSM}} is NP-hard by showing a polynomial-time reduction. To solve the problem, we propose two efficient and effective heuristic algorithms called {\textsf{APA}} and {\textsf{SEA}}. {\textsf{APA}} is a top-down approach and the {\textsf{SEA}} is a bottom-up approach. Finally, we report on extensive experiments using real-world datasets for demonstrating the superiority of our algorithms.

\bibliographystyle{ACM-Reference-Format}
\bibliography{sample-base}

\end{document}